\documentclass[10pt,twocolumn,twoside]{IEEEtran}
\usepackage{times}
\usepackage{amsmath,epsfig}
\usepackage{amssymb}
\usepackage{amsfonts}
\usepackage{array,dcolumn}

\usepackage{algorithm, algorithmic}
\usepackage{subcaption}
\usepackage{psfrag}
\usepackage{balance}
\usepackage{cite}
\usepackage[all]{xy}
\usepackage{dsfont}
\usepackage{hyperref}
\usepackage{url}
\usepackage{xcolor}
\usepackage[nolist]{acronym}
\usepackage{float}
\usepackage{threeparttable}
\usepackage{cancel}
\usepackage[official]{eurosym}
\usepackage{graphicx}
\usepackage[utf8]{inputenc}
\usepackage[english]{babel}
\usepackage{amsthm}
\usepackage[subnum]{cases}
\usepackage[nocomma]{optidef}
\usepackage[normalem]{ulem}
\usepackage{enumitem}

\usepackage[utf8]{inputenc}
\usepackage{pgfplots}
\usepackage{booktabs, adjustbox}
\DeclareUnicodeCharacter{2212}{−}
\usepgfplotslibrary{groupplots,dateplot}
\usetikzlibrary{patterns,shapes.arrows}
\pgfplotsset{compat=newest}
\usetikzlibrary{shapes.geometric, arrows}

\newcommand{\cs}{\texttt{PS}}

\usepackage{filecontents, pgffor}
\usepackage{listings, xcolor}

\usepackage{xcolor, colortbl}
\usepackage{mathtools,amssymb,lipsum, nccmath}
\usepackage{cuted}
\usepackage[normalem]{ulem}

\newtheorem{definition}{Definition}

\newtheorem{lemma}{Lemma}

\usepackage{multirow}
\usepackage{tikz}
\usepackage{pgfplots}
\pgfplotsset{compat=1.10}
\usepackage{gincltex}
\usepgfplotslibrary{fillbetween}
\usetikzlibrary{patterns}

\def\BibTeX{{\rm B\kern-.05em{\sc i\kern-.025em b}\kern-.08em
		T\kern-.1667em\lower.7ex\hbox{E}\kern-.125emX}}

\hypersetup{
	colorlinks = true,
	linkcolor = blue,
	anchorcolor = black,
	citecolor = red,
	filecolor = blue,
	urlcolor = blue
}

\begin{document}
	\title{Goal-Oriented Communications in Federated Learning via Feedback on Risk-Averse Participation}
	\author{Shashi Raj Pandey, Van Phuc Bui, Petar Popovski%
		\\
		Department of Electronic Systems, Aalborg University, Denmark\\
		Emails: \{srp, vpb, petarp\}@es.aau.dk
		\thanks{This work was supported by the Villum Investigator Grant “WATER” from the Velux Foundation, Denmark.}}

	\maketitle
	\begin{abstract}
		We treat the problem of client selection in a Federated Learning (FL) setup, where the learning objective and the local incentives of the participants are used to formulate a goal-oriented communication problem. Specifically, we incorporate the risk-averse nature of participants and obtain a communication-efficient on-device performance, while relying on feedback from the Parameter Server (\texttt{PS}). 
		A client has to decide its transmission plan on when not to participate in FL. This is based on its intrinsic incentive, which is 
		the value of the trained global model upon  participation by this client. Poor updates not only plunge the performance of the global model with added communication cost but also propagate the loss in performance on other participating devices. We cast the relevance of local updates as \emph{semantic information} for developing local transmission strategies, i.e., making a decision on when to ``not transmit". The devices use feedback about the state of the PS and evaluate their contributions in training the learning model in each aggregation period, which eventually lowers the number of occupied connections. Simulation results validate the efficacy of our proposed approach, with up to $1.4\times$ gain in communication links utilization as compared with the baselines.
	\end{abstract}
	
	
	\IEEEpeerreviewmaketitle	
	\section{Introduction}\label{sec:intro}
Semantic communication (SemCom) is pursued to go beyond Shannon's interpretation of reliable transmission as an engineering problem, where all bits are treated equally, and focus on the rationale behind data transmission and its utility at the receiver's end through interpretation of content, context and timing aspects \cite{gunduz2022beyond, kountouris2021semantics, strinati20216g, 10024766}. Consequently, this perspective offers methodologies to limit unnecessary transmissions and prioritize data to improve goal-oriented performance; hence, SemCom is unanimously attributed to being goal-oriented communications. 

Edge intelligence at scale is envisioned to be one of the prominent realizations of the next generation of communication systems, where goal-oriented communications can play a significant role and offer fundamental gains in several network KPIs (for example, in terms of latency, throughput, compute costs, communication costs, energy-efficiency, and so on). Edge intelligence attributes executing machine learning (ML) algorithms at the network edge to offer low-cost, high accuracy, and fast inference for downstream tasks. Satisfying such an objective necessitates methods to train high-quality learning models in a communication-efficient, decentralized manner without tampering with data privacy for heterogeneous services. Federated Learning (FL) \cite{mcmahan2017communication} is one of such popularized paradigms we are interested in this work, where we show a way to enable goal-oriented communications and absorb the corresponding gain in the performance metrics. FL allows training a learning model in a privacy-preserving distributed manner, with a central parameter server (\cs) soliciting devices to share the local learning parameters instead of their private raw data. The iterative model training process is refuelled with an updated global model after the model aggregation at each interaction round -- i.e., the global communication -- between the devices and the \cs. Two regimes of model aggregation exist in the FL operation: (i) Synchronous (SYNC) and (ii) Asynchronous (ASYNC). SYNC imposes tightly coupled iterative training procedures, with issues due to stragglers, leading to poor \emph{training efficiency}. ASYNC relaxes the aggregation period, with issues related to poor model updates, i.e., staleness; hence, poor \emph{communication efficiency}. Therefore, in both regimes, the selection of local models to incorporate in the aggregation process is at the crux of the model training procedure, which, in actuality, contributes to the number of incurred transmissions. 
        \begin{figure}[t!]
            \centering          
            \includegraphics[width=\linewidth]{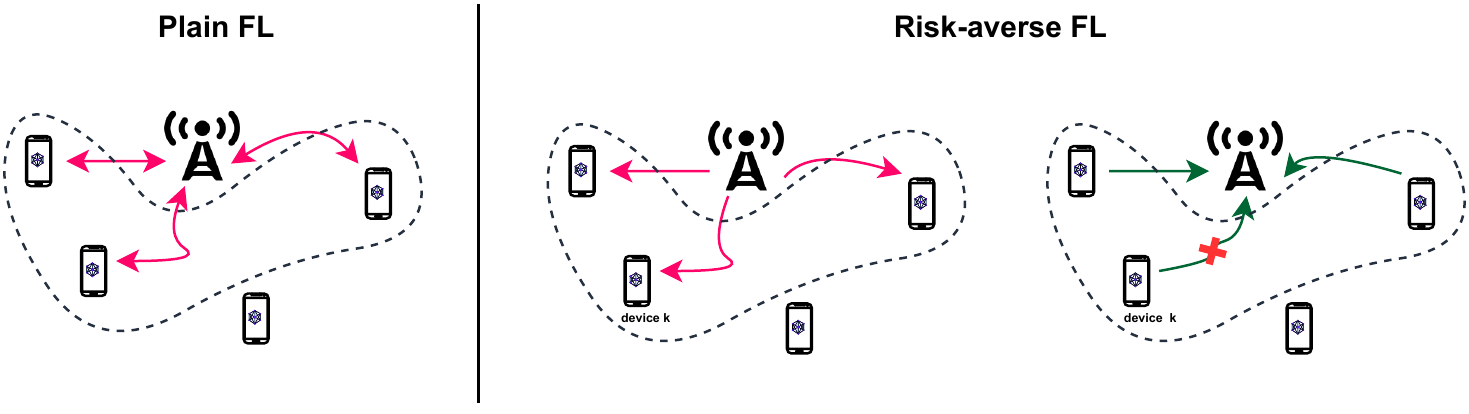}
            \caption{Risk-averse FL setting: (Left figure) Plain FL, and in (Right figure) device $k$ is a risk-averse client, which drops itself from the FL process following \emph{feedback} from \cs at the base station.}
            \label{fig:illustration}
        \end{figure}
        
Device selection problem \cite{kairouz2021advances, nishio2019client, yoshida2020mab}, i.e., activating a subset of devices to participate in the model training, and the design of an efficient model aggregation scheme \cite{pandey2022contribution}, i.e., choosing a subset of obtained local model parameters, has been of key interest in the recent FL research works. However, we argue device selection problem has always been thought of as the \cs's problem but not considered as an opportunity to stimulate strategic participation of the selected devices; hence, ignoring distributed coordination amongst participating devices for goal-oriented communications. Most of the existing research focuses on tuning the aggregation procedure by choosing devices that are valuable for the average model performance of the trained model. However, taking the perspective of a device, the improvements in their learning tasks through arbitrary participation strategy in the FL training can be nominal with an average model performance. This suggests the limited influence of participation of such devices; due to statistical and system-level heterogeneity, not all selected devices in each communication round contribute greatly to the trained model. Then, it is intuitive for devices to decide when not to transmit as well to minimize unnecessary computing and communication costs. We call such devices ``risk-averse'', where risk is the nominal value of the trained model upon participation. In related works, device selection at the \texttt{PS} overlooks the individual risks and the impact on the further participation of each specific device. In actuality, we observe incorporating risk-averse attributes of devices in participation, i.e., risk-averse participation, can address goal-oriented communication problems in FL. We show that it is imperative and of great value to keep track of the performance of the trained model at both \texttt{PS} and the devices, particularly to ensure generalization with improved personalization performance. Therein, feedback from the \texttt{PS} to evaluate risks on participation enables devices to align their transmission strategies while satisfying the overall goal-oriented performances after model aggregation at the \cs.

\subsection{Contributions} 
We pose the following problem of goal-oriented communications in FL: Given the global model and feedback on its instantaneous valuation from the \cs, how does the risk-averse participation of devices impact the number of transmissions? Precisely, we want to understand the perspective of devices in training the FL model with their contributions and the value of the global model for their local learning tasks, i.e., personalization performance. This is done following the proposed \emph{strategic} device selection and participation. The devices derive and evaluate the estimate of the local risk imposed by the trained model and make a participation decision in the FL process, anticipating the value of the investment in training the global model. Our approach is applicable to any closed-loop system and is shown to be effective in limiting unnecessary transmissions and improving goal-oriented performance.

\section{System Model and Problem Definition}\label{sec:sysmodel}
Consider a sensory network with a resourceful \cs and several privately owned resource-constrained IoT devices. Let $\mathcal{K}=[k]$ denote a set of $K$ IoT devices each holding data samples $\mathcal{D}_k$ of size $|\mathcal{D}_k|=D_k$, with each instance of data $x^{(i)} , \ldots, x^{D_k} \in [0,1]^p$ explained by $p$-features in a typical supervised setting, and the set of $D$ total training data samples is $\mathcal{D}$. As the devices are privately owned, we make a fair assumption that the devices are privacy-sensitive, i.e., the data are not shared with the \texttt{PS} unless for some incentives, e.g., monetary compensation for the incurred cost to transmit private data and/or cost of lowering individual privacy concerns. This limits \cs to infer a summary statistics model $\mathbf{w}\in \mathbb{R}^d$ of the sensory network with high confidence and efficiency. 

We formalize the learning problem as
\begin{equation}
    \underset{\mathbf{w}\in \mathbb{R}^d}{\min} \mathcal{L}(\mathbf{w}) \; \; \text{with} \; \; \mathcal{L} (\mathbf{w})  := \frac{1}{D}\sum\nolimits_{i= 1}^D L(x^{(i)}; \mathbf{w}), \label{eq:globalproblem}
\end{equation}
where $\mathbf{w}\in \mathbb{R}^d$ is the unknown model parameter explaining the summary statistics of the distributed data samples; $\mathcal{L}: \mathbb{R}^d \rightarrow \mathbb{R}$ is the empirical risk (i.e., the objective), and $L(x^{(i)}; \mathbf{w}): [0,1]^p \times \mathbb{R}^d \rightarrow \mathbb{R}, \forall i \in \mathcal{D}$ is the loss function associate with data samples at the device. Problem~\eqref{eq:globalproblem} can be solved in a distributed setting as of FL with iterative algorithms leveraging the variants of gradient descent methods, such as stochastic gradient descent (SGD). Then, the data-parallel reformulation of \eqref{eq:globalproblem} is
\begin{equation}
    \underset{\mathbf{w}\in \mathbb{R}^d}{\min} \mathcal{L} (\mathbf{w}) := \frac{1}{D}\sum\nolimits_{k= 1}^K \sum\nolimits_{i= 1}^{D_k} L_k(x^{(i)}; \mathbf{w}), \label{eq:FL_globalproblem}
\end{equation}
where $L_k$ is the local loss function. A basic operation sketch Federated SGD (FedSGD) \cite{mcmahan2017communication} to solve this is discussed below. At each communication round $t$, the \cs broadcast a copy of global model $\mathbf{w}^t$ to the devices of her interest pool (say $K$ devices), either randomly or following a specific strategy; the chosen devices act as a worker on the received global model; compute $\nabla L_k(\mathbf{w}^t)$, and strategically\footnote{We look at distributed coordination of devices in transmitting the local model parameters. We use risk as a semantic measure to capture the relevance of the device's data in FL training.} upload it to the \cs. The round ends once the \cs receives gradients from the devices and executes the model parameters update as
\begin{equation}
    \mathbf{w}^{t+1} = \mathbf{w}^t - \eta \sum\nolimits_{k \in \mathcal{K}} \frac{D_k}{D} \nabla L_k(\mathbf{w}^t), \label{eq:gd}
\end{equation}
where $\eta$ is a step size. Note that the aggregation of gradients in \eqref{eq:gd} can be modified, and the update strategy varies following design requirements, for instance, with multiple steps of GD during local training, as Federated Averaging (FedAvg) in \cite{mcmahan2017communication}. 

For this, the overall learning objective, set as a goal-oriented communications problem in FL, at the \cs, is to train a high-quality learning model with better generalization performance in a communication-efficient manner. This boils down to an iterative procedure of (i) efficient device selection for \cs -, with added feedback that allows (ii) strategic response of the selected devices to tune their transmission strategies that lower the participation risk, hence, exploiting the relevance of data contribution in FL training. In the following, we systematically present the proposed solution approach.

\section{Risk-averse device selection with feedback}
Let $\mathcal{A}^t_k \in \{0,1\}$ be the control algorithm for device $k$ denoting participation strategy in the model training process at time instance $t>0$. Then, each device is endowed with an instantaneous valuation function $v^t_k:\mathbb{R}^d \rightarrow \mathbb{R}$, where $v^t_k(\mathbf{w}^t|\mathcal{A}^t_k =1)$ denotes the value of device $k$ for model $\mathbf{w}$. In an iterative model training procedure given the time horizon $t\in\{1,2,\ldots,\tau\}$, such as in FL, the \cs supplies copies of the updated global model in every communication round $t$\footnote{A communication round in FL is defined as the interaction call between \cs and the involved devices for the exchange of learned model parameters, such as gradients and weights.}. Then, without loss of generality, we define the average valuation of device $k$ in the training procedure as $\frac{1}{\tau}\sum\nolimits_{t=1}^{\tau}v^t_k(\mathbf{w}^t|\mathcal{A}^t_k =1)$. 

\textbf{Local Regret function: } We use regret function $\mathcal{R}_k(\mathbf{w}^t):\mathbb{R}^d \rightarrow [0,1]$ to define drop in the valuation of the device's model parameters for building the global model $\mathbf{w}$ in every successive communication round. Formally, $\mathcal{R}_k(\mathbf{w}^t) = v^t_k(\mathbf{w}^t) - v^{t-1}_k(\mathbf{w}^{t-1})$, with $v^0_k(\mathbf{w}^0)$ as initial valuation of device $k$ of model $\mathbf{w}^0$ broadcast by the \cs. We will assume the devices drop off the training process at any $t$ such that $\mathcal{R}_k(\mathbf{w}^t)\le \delta_k\footnote{This parameter captures the system-level constraints on each device.}$, i.e., the communication round, after which the global model negatively influences the local performance. Note that this stopping time is private information, and the device contribution might be valuable for the \texttt{PS} to improve generalization.

\textbf{Multi-arm Bandits (MAB)-based Selection:} In each communication round, the \cs selects $\mathcal{M} \subseteq \mathcal{K}$ devices, where choosing each device is considered as an arm pulled for an obtained valuation, i.e., the instantaneous reward, plus a penalty factor to capture the number of times that arm has been pulled. 
Upper Confidence Bound (UCB) algorithm assigns the following value to each arm:
\begin{equation}
   f_{k, t} = \underbrace{\mu_{k,t}}_{\text{average marginal contribution}} + \underbrace{\sqrt{\frac{\ln t}{N_{k,t}}}}_{\text{penalty term}}, \label{eq:ucb}
\end{equation}
where $N_{k,t}$ is the number of times device $k$ has been selected before. Then, selecting an arm is equivalently solving the following:
\begin{equation}
   \mathcal{M}_{t} = \underset{k\in\{1,2,\ldots,K(t)\}}{\arg\max} f_{k, t},\label{eq:ucb_select}
\end{equation}
where $K(t)$ is the number of possible devices available at time instance $t$, and $\mathcal{M}_t$ is the set of $M_t$ devices selected amongst $K(t)$ choices, i.e., $\mathcal{M}_t \in\{1,2,\ldots,K(t)\}$. Since the \cs select a set of devices $\mathcal{M}_t$ per communication round, it naively applies the UCB score \eqref{eq:ucb}, orders possible candidate devices based on their contribution, and selects a subset of them per the available candidate selection budget to define $M_t$.

In \eqref{eq:ucb}, it is straightforward for the \texttt{PS} to evaluate past device selection choices; however, the value of instantaneous marginal contribution  $\mu_{k,t}$ is \emph{unknown} to the devices; hence, they naively transmit the local model parameters when asked for. At the \texttt{PS} side, we resort to Shapley Value (SV) \cite{shapley1953value} for evaluating the contribution of the received model parameters. In principle, SV calculates the average marginal contribution of each device $k$ on all possible combinations of coalition subsets as
\begin{equation}
    \mu_{k,t}(K, U) = \frac{1}{K!} \sum_{\mathcal{S} \subseteq \mathcal{K} \setminus \{k\}} \frac{U(\mathcal{S}\cup \{k\}) - U(\mathcal{S})}{\binom{K-1}{|\mathcal{S}|}},
    \label{eq:standard shapely}
\end{equation}
where the function $U(\cdot)$ gives the value for any subset of the devices, e.g., let $\mathcal{S}$ be a subset of $K$, then $U(\mathcal{S})$ gives the value of that subset \cite{pandey2022fedtoken}. For devices, in literature, ways such as the statistical methods and game-theoretic models \cite{agussurja2022convergence, jia2019towards} exist to use for an estimation of $\mu_{k,t}$ in each aggregation round $t$, though not considering the overall objective of the \cs. 

We argue \texttt{PS} can incorporate one of these techniques to maximize its expected long-term reward with the virtue of an efficient client selection by the time horizon while signalling the devices with feedback to infer the relevance of their transmission in prior. This is done as follows. Along with the aggregated global model, the \texttt{PS} feedback the model regret (defined later in \eqref{eq:cvar}) due to risk-averse participation in the successive communication round. The feedback is provided to specific devices; these devices are strategically selected by executing the UCB-based algorithm while considering the risk-averse participation of the devices. To formulate such device selection, the \texttt{PS} first sniffs the average risk across devices participating in the FL process. This allows it to maximize the benefits of its selection strategy in the subsequent communication round, i.e., derive $M_t$ for the global model broadcast and provide the feedback. After local training, the selected devices evaluate the utility of participation as the value of the trained model with the feedback and derive the relevance of their transmitted model parameters at the \texttt{PS} in the next round. This suggests an effective approach for goal-oriented communications in FL. Therein, we have the following definition.
  	\begin{algorithm}[t!]
        	\caption{\strut Risk-averse Participation with Feedback in FL}
        	\label{alg:risk}
        	\begin{algorithmic}[1]
        		\STATE{Initialize: Number of devices $K(t)$, $\delta_k, \forall k \in [k]$}.
        		 \REPEAT
                     \STATE {Execute UCB algorithm that solves \eqref{eq:ucb_select};}
        		   \STATE {Broadcast feedback $\{ \mathbf{w}^t, \mathcal{R}_{t, \theta}\}$};\\
        		    \FORALL{devices $m\in M_t$}
                    \STATE {Set $v_m^{(t+1)} \leftarrow \mathcal{R}_{t, \theta}$; }
                    \STATE {Solve the local learning problem using SGD};
                    \STATE {Evaluate $\mathcal{R}_m(\mathbf{w}^t)$}
                    \STATE {Update transmission variable $\mathcal{A}^t_m$};
                    \ENDFOR
        	     \UNTIL{the model converges.}
        	\end{algorithmic}
        	\label{Algorithm}
        \end{algorithm}
\begin{definition}\label{def:participation}
Consider a dynamic population of devices in the FL system. Then, if $\theta \in [0,1]$ is the average risk across participating devices, $K_{\theta}(\epsilon, t)$ is the number of devices having data that can contribute at least $\epsilon\in[0,1]$ improvement in the learning model.
\end{definition}
Following Definition~\ref{def:participation}, the available devices with average $\theta$ risk on the global learning model can be obtained as
\begin{equation}
    K_{\theta}(t) = \underset{\epsilon \rightarrow 1}{\lim}  \; K_{\theta}(\epsilon, t).
\end{equation}
Then, the total available devices at the communication round $t$ are
\begin{equation}
    K(t) = \sum\nolimits_{\theta}K_\theta(t),
\end{equation}
with the average risk of collaborative training defined as
\begin{equation}
    R(t) = \frac{ \sum\nolimits_{\theta}\theta \cdot K_\theta(t)}{K(t)}.
\end{equation}
Note that, in this work, we are not interested in modelling the arrival process of the local parameters. Instead, our aim is to develop a device selection strategy at the \texttt{PS} and the participation strategy at the risk-averse device's side with the feedback.
Definition~\ref{def:participation} considers average $\theta \in [0,1]$ risk across the participating devices that brings a value of averaged reward $[0, 1]$,  and plans also to empower the performance of risk-averse devices for better generalization. To be \emph{semantically relevant}, devices should also be strategic in updating the local parameters while the aggregator implementation of UCB for candidate selection in every global iteration. With this motivation, we adopt value-at-risk (VaR) in reformulating the stochastic bandit selection criterion to perform device selection at the \cs, which introduces extra signalling for the devices to align their transmission plan. We begin defining conditional VaR (CVaR) in the following.

For any bounded random variable $X$ with cumulative distribution function (CDF) $F(x) = \mathbb{P}[X\le x]$, the Conditional VaR (CVaR) of $X$ is defined as \cite{rockafellar2000optimization}
\begin{equation}
    \textrm{CVaR}_\alpha(X) := \underset{\nu}{\textrm{sup}}\bigg\{\nu - \frac{1}{\alpha}\mathbb{E}[(\nu - X)^+]\bigg\},
\end{equation}
at level $\alpha \in (0,1)$ signifying worst $\alpha-$ fraction of realizations, where $(x)^+ = \max(0, x)$. CVaR analysis allows the \texttt{PS} to consider the full distribution of the contribution of devices while making strategic device selection in the FL training process, given devices are risk-sensitive. It also captures the average drop in the learning performance following this naive device selection approach. We recast the device-selection problem using CVaR-based UCB as follows. Let $F_{\mu_k}$ be the CDF of the contribution selecting device $k$ on estimated $[0, \theta]$ capturing population risk. The true distribution of device contribution is obtained following the available samples in each round and evaluating marginal contributions of devices in training the global model. With the action profile on device participation upon random selection $\mathcal{A}_t$ with probability $\varepsilon$ at time $t$, i.e., selecting a subset of devices from a pool of $K(t)$ devices, we then use the definition of CVaR-regret at global aggregation round $\tau$ as 
\begin{equation}
    \mathcal{R}_{\tau, \theta} = \tau \max_k \bigg[\textrm{CVaR}_\theta(F_{\mu_k})\bigg] - \mathbb{E}\bigg\{ \sum\nolimits_{t=1}^{\tau}\textrm{CVaR}_\theta(F_{\mathcal{A}_t})\bigg\},
    \label{eq:cvar}
\end{equation}
to capture the contribution profile of selected devices. The right-hand term, considering the most optimistic estimates, can be approximated to be bounded as,
\begin{equation}
    \mathcal{R}_{\tau, \theta} \le \frac{4 \max\{\epsilon_k, \forall k \in \mathcal{A}_t\}}{\theta}\sqrt{\tau K \ln(\sqrt{2}\tau)},
    \label{eq:cvar_regret}
\end{equation}
where $\theta$ is updated in each global communication round with the resulting CVaR-based risk-averse device selection problem as
\begin{equation}
   \mathcal{M}_{t\in \{1,2,\ldots,\tau\}} = \underset{k\in\{1,2,\ldots,K(t)\}}{\arg\max} f_{k, t} - \mathbb{E}[\mathcal{R}_{t, \theta}].\label{eq:ucb_cvar_select}
\end{equation}
The derivation follows [\cite{dvoretzky1956asymptotic}, Theorem 1. \cite{tamkin2019distributionally}].
Algorithm~\ref{Algorithm}. summarizes the proposed approach of strategic selection and risk-averse participation to meet the learning objectives, i.e., improving training and learning efficiency through feedback; hence, goal-oriented communications in FL. Therein, we have the following lemma that quantifies the analytical bounds on stability and convergence of the average risk, as reflected by the average per-device model performance.

\begin{lemma}
Let $\theta_1, \theta_2, \ldots, \theta_k$ be independent random variables denoting the quality of local risk at device $k \in \mathcal{K}$ in any instance such that $|\theta_k| \le \theta_{\textrm{th}}$ almost surely, and $\mathbb{E}[\theta_k^2] \le \alpha_k$ for every $k$. Then, for every $\alpha \ge 0$ we have
\begin{equation}
    \mathsf{Pr}(|\theta - \mathbb{E}[\theta]| \ge \epsilon) \le \mathsf{exp}\bigg(-\frac{\epsilon^2}{2(A + \frac{\epsilon\theta_{\textrm{th}}}{3})}\bigg),
\end{equation}
where $\theta = \sum\nolimits_{k=1}^{K} \theta_k$ and $A = \sum\nolimits_{k=1}^{K}\epsilon_k$.
\end{lemma}
\begin{proof}
The proof is easy to establish following Bernstein's inequality \cite{hall1967measures}.
\end{proof}

	\section{Performance Evaluation}
\subsection{Settings}
\subsubsection{Dataset and Model}
To validate the efficacy of our proposed approach, we conducted experiments with the well-known MNIST dataset \cite{lecun1998gradient}, covering $70000$ data samples of $28X28$ grey images in 10 classes equally. The \texttt{PS} is training a convolutional neural network (CNN) model in a distributed manner. The CNN model has two hidden layers for feature extraction (10 filters for the first layer and 20 filters for the second layer with kernel size as $5\times5$) and 2 fully connected layers for classification.
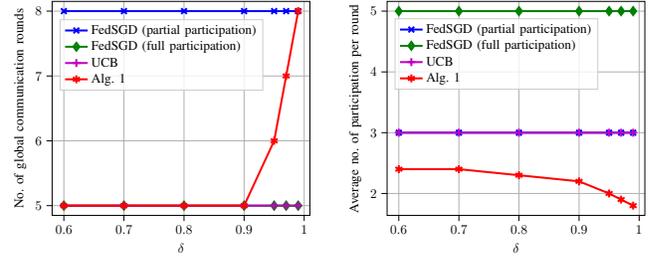
\begin{figure}[t!]
\begin{subfigure}[b]{0.20\textwidth}
\centering
\begin{tikzpicture}[scale=0.5]

\definecolor{darkgray176}{RGB}{176,176,176}
\definecolor{darkviolet1910191}{RGB}{191,0,191}
\definecolor{green01270}{RGB}{0,127,0}
\definecolor{lightgray204}{RGB}{204,204,204}

\begin{axis}[
legend cell align={left},
legend style={
  fill opacity=0.8,
  draw opacity=1,
  text opacity=1,
  at={(0.03,0.75)},
  anchor=west,
  draw=lightgray204
},
tick align=outside,
tick pos=left,
x grid style={darkgray176},
xlabel={\(\displaystyle \delta\)},
xmajorgrids,
xmin=0.5805, xmax=1.0095,
xtick style={color=black},
xticklabel style={},
y grid style={darkgray176},
ylabel={No. of global communication rounds},
ymajorgrids,
ymin=4.85, ymax=8.15,
ytick style={color=black}
]
\addplot [line width=1.5, blue, mark=x, mark size=3, mark options={solid}]
table {%
0.6 8
0.7 8
0.8 8
0.9 8
0.95 8
0.97 8
0.99 8
};
\addlegendentry{FedSGD (partial participation)}
\addplot [line width=1.5, green01270, mark=diamond*, mark size=3, mark options={solid}]
table {%
0.6 5
0.7 5
0.8 5
0.9 5
0.95 5
0.97 5
0.99 5
};
\addlegendentry{FedSGD (full participation)}
\addplot [line width=1.5, darkviolet1910191, mark=+, mark size=3, mark options={solid}]
table {%
0.6 5
0.7 5
0.8 5
0.9 5
0.95 5
0.97 5
0.99 5
};
\addlegendentry{UCB}
\addplot [line width=1.5, red, mark=asterisk, mark size=3, mark options={solid}]
table {%
0.6 5
0.7 5
0.8 5
0.9 5
0.95 6
0.97 7
0.99 8
};
\addlegendentry{Alg. 1}
\end{axis}
\end{tikzpicture}

\caption{Total communication rounds.}
\label{fig:1}
\end{subfigure}\hspace{2em}
\begin{subfigure}[b]{0.2\textwidth}
\centering
\begin{tikzpicture}[scale=0.5]

\definecolor{darkgray176}{RGB}{176,176,176}
\definecolor{darkviolet1910191}{RGB}{191,0,191}
\definecolor{green01270}{RGB}{0,127,0}
\definecolor{lightgray204}{RGB}{204,204,204}

\begin{axis}[
legend cell align={left},
legend style={
  fill opacity=0.8,
  draw opacity=1,
  text opacity=1,
  at={(0.03,0.92)},
  anchor=north west,
  draw=lightgray204
},
tick align=outside,
tick pos=left,
x grid style={darkgray176},
xlabel={\(\displaystyle \delta\)},
xmajorgrids,
xmin=0.5805, xmax=1.0095,
xtick style={color=black},
xticklabel style={},
y grid style={darkgray176},
ylabel={Average no. of participation per round},
ymajorgrids,
ymin=1.64, ymax=5.16,
ytick style={color=black}
]
\addplot [line width=1.5, blue, mark=x, mark size=3, mark options={solid}]
table {%
0.6 3
0.7 3
0.8 3
0.9 3
0.95 3
0.97 3
0.99 3
};
\addlegendentry{FedSGD (partial participation)}
\addplot [line width=1.5, green01270, mark=diamond*, mark size=3, mark options={solid}]
table {%
0.6 5
0.7 5
0.8 5
0.9 5
0.95 5
0.97 5
0.99 5
};
\addlegendentry{FedSGD (full participation)}
\addplot [line width=1.5, darkviolet1910191, mark=+, mark size=3, mark options={solid}]
table {%
0.6 3
0.7 3
0.8 3
0.9 3
0.95 3
0.97 3
0.99 3
};
\addlegendentry{UCB}
\addplot [line width=1.5, red, mark=asterisk, mark size=3, mark options={solid}]
table {%
0.6 2.4
0.7 2.4
0.8 2.3
0.9 2.2
0.95 2
0.97 1.9
0.99 1.8
};
\addlegendentry{Alg. 1}
\end{axis}
\end{tikzpicture}
\caption{Communication links per communication round.}
\label{fig:2}   
\end{subfigure}
\vspace{0.2cm}
\caption{Comparative performance evaluation to reach $80\%$ model accuracy for varying $\delta$.} 
\label{fig:com_cost}
\end{figure}
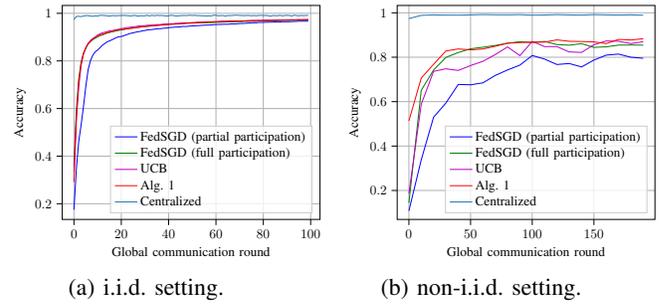
\begin{figure}[t!]
\begin{subfigure}[b]{0.20\textwidth}
\centering
\begin{tikzpicture}[scale=0.5]

\definecolor{darkgray176}{RGB}{176,176,176}
\definecolor{darkviolet1910191}{RGB}{191,0,191}
\definecolor{green01270}{RGB}{0,127,0}
\definecolor{lightgray204}{RGB}{204,204,204}
\definecolor{steelblue31119180}{RGB}{31,119,180}

\begin{axis}[
legend cell align={left},
legend style={
  fill opacity=0.8,
  draw opacity=1,
  text opacity=1,
  at={(0.97,0.03)},
  anchor=south east,
  draw=lightgray204
},
tick align=outside,
tick pos=left,
x grid style={darkgray176},
xlabel={Global communication round},
xmajorgrids,
xmin=-4.95, xmax=103.95,
xtick style={color=black},
y grid style={darkgray176},
ylabel={Accuracy},
ymajorgrids,
ymin=0.134035, ymax=1.033065,
ytick style={color=black}
]
\addplot [semithick, blue]
table {%
0 0.1749
1 0.3435
2 0.4601
3 0.5295
4 0.6062
5 0.6924
6 0.7633
7 0.8019
8 0.8246
9 0.8384
10 0.846
11 0.8558
12 0.866
13 0.8727
14 0.8803
15 0.8842
16 0.8884
17 0.8937
18 0.8983
19 0.9017
20 0.9025
21 0.9058
22 0.9091
23 0.9129
24 0.9161
25 0.918
26 0.9205
27 0.9227
28 0.9239
29 0.9268
30 0.9278
31 0.9302
32 0.9311
33 0.9319
34 0.9333
35 0.9347
36 0.9355
37 0.936
38 0.9377
39 0.9377
40 0.9395
41 0.9408
42 0.9412
43 0.9424
44 0.9424
45 0.9428
46 0.9454
47 0.9445
48 0.9454
49 0.9464
50 0.9469
51 0.9475
52 0.9481
53 0.9487
54 0.9495
55 0.9501
56 0.9504
57 0.9505
58 0.9514
59 0.9524
60 0.9527
61 0.9534
62 0.9531
63 0.954
64 0.9541
65 0.9541
66 0.9552
67 0.9558
68 0.9555
69 0.9564
70 0.957
71 0.9576
72 0.9577
73 0.9586
74 0.9591
75 0.96
76 0.9605
77 0.9601
78 0.9618
79 0.9616
80 0.962
81 0.9625
82 0.9628
83 0.9628
84 0.9632
85 0.9641
86 0.9632
87 0.9636
88 0.9644
89 0.9644
90 0.9646
91 0.9648
92 0.9653
93 0.9658
94 0.9661
95 0.9671
96 0.9669
97 0.9672
98 0.967
99 0.9671
};
\addlegendentry{FedSGD (partial participation)}
\addplot [semithick, green01270]
table {%
0 0.3601
1 0.5916
2 0.7231
3 0.7948
4 0.8315
5 0.8542
6 0.8712
7 0.8813
8 0.8895
9 0.8954
10 0.9008
11 0.9064
12 0.9093
13 0.9128
14 0.9172
15 0.9192
16 0.9222
17 0.9242
18 0.9266
19 0.9283
20 0.9305
21 0.9315
22 0.9345
23 0.9362
24 0.9372
25 0.9383
26 0.9396
27 0.9403
28 0.9424
29 0.9451
30 0.9442
31 0.9456
32 0.9471
33 0.9479
34 0.949
35 0.9497
36 0.9508
37 0.9523
38 0.9523
39 0.9533
40 0.9545
41 0.9551
42 0.9551
43 0.9557
44 0.956
45 0.9569
46 0.9575
47 0.9576
48 0.9584
49 0.9587
50 0.9594
51 0.9594
52 0.9597
53 0.9601
54 0.961
55 0.9607
56 0.9603
57 0.9611
58 0.9607
59 0.9617
60 0.962
61 0.9626
62 0.9621
63 0.9631
64 0.9635
65 0.9644
66 0.9643
67 0.9644
68 0.9646
69 0.9653
70 0.9651
71 0.965
72 0.9664
73 0.966
74 0.9661
75 0.9668
76 0.9664
77 0.9664
78 0.9674
79 0.9681
80 0.9679
81 0.9688
82 0.9686
83 0.969
84 0.9689
85 0.9684
86 0.9692
87 0.9693
88 0.9702
89 0.9706
90 0.9702
91 0.97
92 0.9708
93 0.9706
94 0.9709
95 0.9706
96 0.9711
97 0.9712
98 0.9718
99 0.9712
};
\addlegendentry{FedSGD (full participation)}
\addplot [semithick, darkviolet1910191]
table {%
0 0.3341
1 0.5671
2 0.6864
3 0.7805
4 0.8294
5 0.8531
6 0.8726
7 0.887
8 0.8947
9 0.8999
10 0.9088
11 0.9139
12 0.9183
13 0.9203
14 0.9231
15 0.9262
16 0.9269
17 0.9296
18 0.9315
19 0.9334
20 0.9361
21 0.9371
22 0.9395
23 0.9408
24 0.9421
25 0.9432
26 0.9444
27 0.9464
28 0.947
29 0.9475
30 0.948
31 0.9499
32 0.9501
33 0.9509
34 0.9508
35 0.9528
36 0.9523
37 0.9531
38 0.9549
39 0.9551
40 0.9559
41 0.9562
42 0.9564
43 0.9579
44 0.9575
45 0.9587
46 0.9595
47 0.96
48 0.9602
49 0.9601
50 0.9614
51 0.9613
52 0.9624
53 0.9621
54 0.9627
55 0.9634
56 0.9638
57 0.9641
58 0.9643
59 0.9647
60 0.9652
61 0.9653
62 0.9664
63 0.9661
64 0.9662
65 0.967
66 0.9671
67 0.9678
68 0.9676
69 0.9672
70 0.9681
71 0.9684
72 0.9683
73 0.9684
74 0.9687
75 0.9692
76 0.9688
77 0.9689
78 0.9694
79 0.9694
80 0.9704
81 0.9703
82 0.971
83 0.9704
84 0.9707
85 0.971
86 0.9709
87 0.9711
88 0.9712
89 0.9712
90 0.9721
91 0.9718
92 0.9725
93 0.9719
94 0.9724
95 0.9722
96 0.9727
97 0.9725
98 0.9729
99 0.9733
};
\addlegendentry{UCB}
\addplot [semithick, red]
table {%
0 0.2911
1 0.5411
2 0.7117
3 0.7846
4 0.8292
5 0.8554
6 0.8741
7 0.882
8 0.8929
9 0.8998
10 0.9059
11 0.9089
12 0.9125
13 0.9155
14 0.9189
15 0.9209
16 0.9243
17 0.9257
18 0.9299
19 0.9303
20 0.9323
21 0.9332
22 0.9343
23 0.9343
24 0.9369
25 0.9379
26 0.9398
27 0.9409
28 0.942
29 0.9433
30 0.9449
31 0.9455
32 0.946
33 0.947
34 0.9483
35 0.9489
36 0.9496
37 0.9503
38 0.9507
39 0.9517
40 0.9531
41 0.9535
42 0.9542
43 0.9548
44 0.9557
45 0.9561
46 0.957
47 0.9579
48 0.9594
49 0.9595
50 0.9595
51 0.9601
52 0.9611
53 0.9617
54 0.9616
55 0.9624
56 0.9634
57 0.9637
58 0.9637
59 0.9646
60 0.9645
61 0.9654
62 0.966
63 0.9667
64 0.9667
65 0.9666
66 0.9669
67 0.968
68 0.9679
69 0.9672
70 0.9678
71 0.9681
72 0.968
73 0.9693
74 0.9694
75 0.9698
76 0.9702
77 0.9701
78 0.9704
79 0.9699
80 0.9705
81 0.9709
82 0.9712
83 0.9708
84 0.971
85 0.9713
86 0.9725
87 0.9719
88 0.9724
89 0.9723
90 0.9728
91 0.9728
92 0.9732
93 0.9731
94 0.9732
95 0.9742
96 0.9731
97 0.9741
98 0.974
99 0.9744
};
\addlegendentry{Alg. 1}
\addplot [semithick, steelblue31119180]
table {%
0 0.974
1 0.9871
2 0.9875
3 0.9863
4 0.9885
5 0.9903
6 0.9904
7 0.9886
8 0.9873
9 0.9897
10 0.9884
11 0.9916
12 0.9913
13 0.9906
14 0.9898
15 0.9883
16 0.9903
17 0.9902
18 0.9896
19 0.9908
20 0.9906
21 0.9914
22 0.9906
23 0.9902
24 0.9888
25 0.9908
26 0.991
27 0.9883
28 0.9898
29 0.9895
30 0.9897
31 0.9917
32 0.9888
33 0.9916
34 0.9896
35 0.9917
36 0.9907
37 0.9912
38 0.99
39 0.989
40 0.9894
41 0.9918
42 0.9908
43 0.9889
44 0.9891
45 0.9905
46 0.9909
47 0.9906
48 0.991
49 0.9921
50 0.9907
51 0.9918
52 0.9905
53 0.9904
54 0.9898
55 0.9896
56 0.989
57 0.9896
58 0.9891
59 0.9919
60 0.9921
61 0.9897
62 0.9904
63 0.9906
64 0.9884
65 0.9908
66 0.9905
67 0.9909
68 0.9922
69 0.9919
70 0.9912
71 0.9908
72 0.9905
73 0.9911
74 0.991
75 0.9898
76 0.991
77 0.9896
78 0.9911
79 0.9913
80 0.9909
81 0.991
82 0.9905
83 0.9877
84 0.9907
85 0.9902
86 0.9898
87 0.9918
88 0.9895
89 0.9908
90 0.9917
91 0.991
92 0.9888
93 0.9922
94 0.9907
95 0.9897
96 0.9911
97 0.9903
98 0.9911
99 0.9913
};
\addlegendentry{Centralized}
\end{axis}

\end{tikzpicture}
\caption{i.i.d. setting.}
\label{fig:lp1}
\end{subfigure}\hspace{2em}
\begin{subfigure}[b]{0.2\textwidth}
\centering
\begin{tikzpicture}[scale=0.5]

\definecolor{darkgray176}{RGB}{176,176,176}
\definecolor{darkviolet1910191}{RGB}{191,0,191}
\definecolor{green01270}{RGB}{0,127,0}
\definecolor{lightgray204}{RGB}{204,204,204}
\definecolor{steelblue31119180}{RGB}{31,119,180}

\begin{axis}[
legend cell align={left},
legend style={
  fill opacity=0.8,
  draw opacity=1,
  text opacity=1,
  at={(0.97,0.03)},
  anchor=south east,
  draw=lightgray204
},
tick align=outside,
tick pos=left,
x grid style={darkgray176},
xlabel={Global communication round},
xmajorgrids,
xmin=-9.5, xmax=199.5,
xtick style={color=black},
y grid style={darkgray176},
ylabel={Accuracy},
ymajorgrids,
ymin=0.069605, ymax=1.033065,
ytick style={color=black}
]
\addplot [semithick, blue]
table {%
0 0.10835
10 0.33785
20 0.5298
30 0.5938
40 0.67735
50 0.6758
60 0.68455
70 0.71795
80 0.74275
90 0.76455
100 0.80785
110 0.7917
120 0.76725
130 0.7716
140 0.7563
150 0.78715
160 0.80945
170 0.8141
180 0.79985
190 0.79525
};
\addlegendentry{FedSGD (partial participation)}
\addplot [semithick, green01270]
table {%
0 0.1444
10 0.65085
20 0.7425
30 0.79915
40 0.82115
50 0.838
60 0.84495
70 0.85305
80 0.86295
90 0.8698
100 0.86685
110 0.87115
120 0.8571
130 0.8542
140 0.86115
150 0.8442
160 0.84685
170 0.8545
180 0.85505
190 0.85445
};
\addlegendentry{FedSGD (full participation)}
\addplot [semithick, darkviolet1910191]
table {%
0 0.18615
10 0.5881
20 0.7363
30 0.7479
40 0.74045
50 0.76245
60 0.78155
70 0.81195
80 0.84645
90 0.80715
100 0.8712
110 0.8479
120 0.84785
130 0.82385
140 0.8217
150 0.8564
160 0.8732
170 0.87195
180 0.8614
190 0.8697
};
\addlegendentry{UCB}
\addplot [semithick, red]
table {%
0 0.5134
10 0.70605
20 0.76825
30 0.8281
40 0.8372
50 0.8333
60 0.83705
70 0.8493
80 0.8632
90 0.866
100 0.8677
110 0.868
120 0.8781
130 0.87175
140 0.87085
150 0.869
160 0.86155
170 0.8801
180 0.87765
190 0.88325
};
\addlegendentry{Alg. 1}
\addplot [semithick, steelblue31119180]
table {%
0 0.974
10 0.9884
20 0.9906
30 0.9897
40 0.9894
50 0.9907
60 0.9921
70 0.9912
80 0.9909
90 0.9917
100 0.9893
110 0.9902
120 0.9917
130 0.9905
140 0.9897
150 0.9919
160 0.9907
170 0.99
180 0.9908
190 0.9891
};
\addlegendentry{Centralized}
\end{axis}

\end{tikzpicture}
\caption{non-i.i.d. setting.}
\label{fig:lp_2}   
\end{subfigure}
\vspace{0.2cm}
\caption{Performance analysis in terms of testing accuracy versus the number of communication rounds.} 
\label{fig:learning_performance}
\end{figure}
\subsubsection{Implementation details} 
The assessment of goal-oriented performance in communication and training efficiency at the central server \texttt{PS} involves analyzing the independent and identically distributed (i.i.d.) and non-i.i.d. characteristics of the used MNIST dataset. We consider the availability of $K=50$ devices, with replication of the data distribution and non-i.i.d. setting for performance evaluation using the methodology outlined in \cite{pandey2022fedtoken}. For training, each device model is trained for one local epoch by the optimizer, with a batch size of 30 and a learning rate of 0.001. We utilized the SGD optimizer for the i.i.d. MNIST dataset and the Adam optimizer for the non-i.i.d. setting to enhance the overall performance. For device selection, we set the selection probability $\varepsilon = 0.1$.

\subsection{Results and Analysis}
We evaluate and compare the performance of our proposed algorithm with four other intuitive baselines: (i) \textit{Centralized}, where all the training samples are available at the \cs; (ii) \textit{FedSGD} \cite{mcmahan2017communication} (with \emph{full} participation); (iii) \textit{FedSGD} (with \emph{partial} participation), where $60\%$ of the devices are selected, and (iv) \textit{UCB} \cite{yoshida2020mab}, where devices are selected purely on UCB metric without considering risk-averse characteristics of the devices.

In Fig.~\ref{fig:com_cost}, we show the performance of our proposed approach in relation to the number of global communication rounds needed to achieve an 80\% model accuracy and the average number of communication links utilized per communication round within the first ten rounds. As observed in Fig.~\ref{fig:1} and \ref{fig:2}, UCB offers a gain in the average communication links in comparison with the traditional FedSGD scheme while requiring the same number of communication rounds to achieve $80\%$ of accuracy. This is expected, as UCB allows strategic selection of devices following their marginal contributions in improving the model performance. Meanwhile, the proposed Alg.~\ref{Algorithm} incorporates risk-averse participation and further improves the performance in terms of incurred number of communication links per communication round as compared to the UCB approach while ensuring the model accuracy, especially with $\delta \leq 0.9$. With $\delta = 0.9$, Alg.~1 exhibited an average link consumption of 2.2, which is 1.4$\times$ lower than UCB and FedSGD techniques, while requiring the same five rounds to achieve $80\%$ model accuracy. This suggests a significant reduction in the number of transmissions to satisfy the KPIs with goal-oriented communications in FL training. 

\begin{figure}
    \centering
    \includegraphics[width=0.4\textwidth]{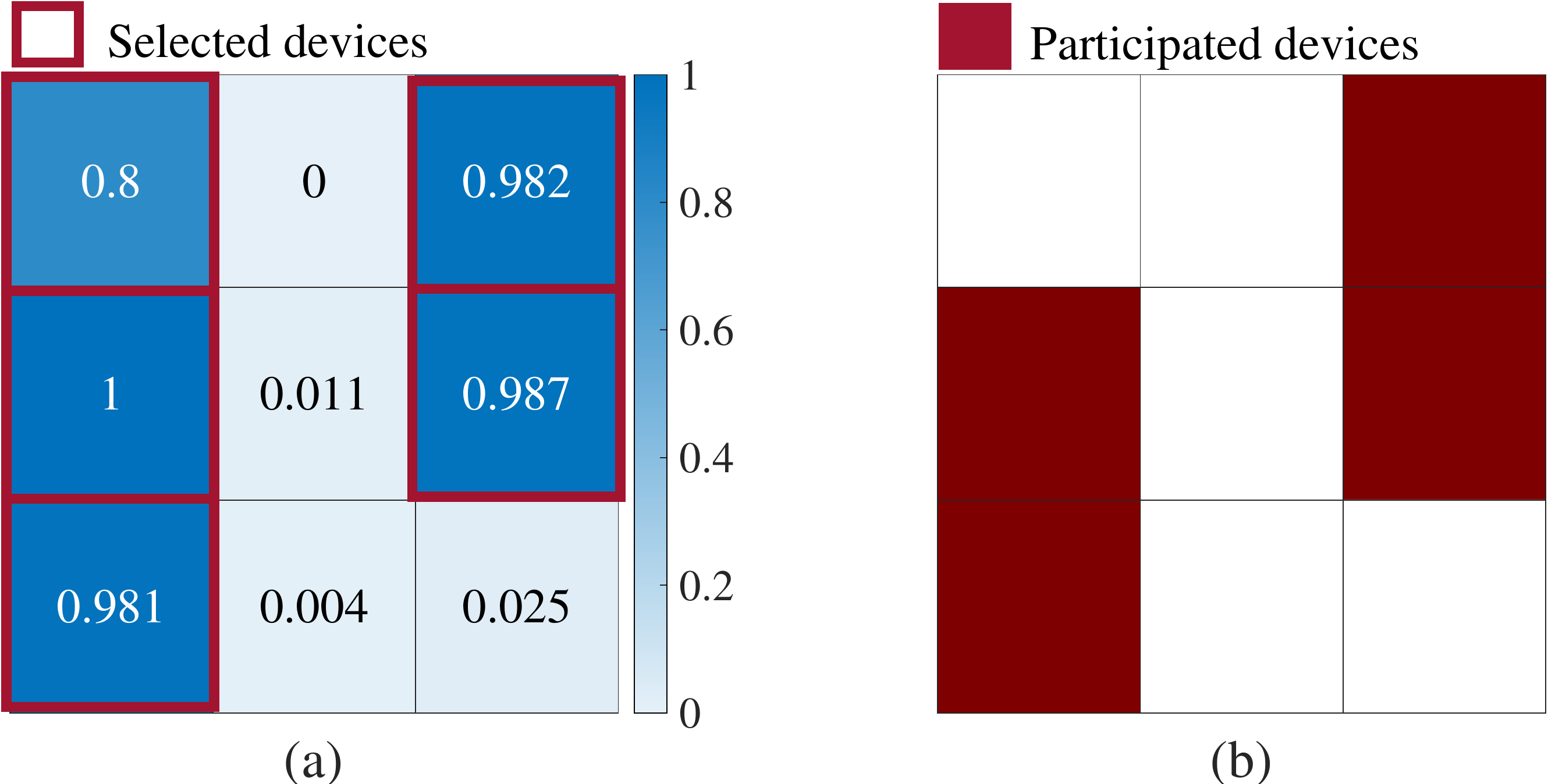}
    \caption{Snapshot of risk-averse device selection following their contributions and participation with $|\mathcal{A}_t|= 9$ at $t=10$. }
    \label{fig:snapshot}
\end{figure}
Fig.~\ref{fig:learning_performance} demonstrates the learning performance in terms of convergence of testing accuracy under both i.i.d and non-i.i.d data settings. We observe Alg.~\ref{Algorithm} outperforms the other baselines while keeping a competitive performance with UCB (both in partial and full participation schemes) with fewer connections per communication round. This analysis is consistent with our observations in Fig.~\ref{fig:com_cost}. Obviously, for the non-i.i.d. setting, we see a drop in the performance; however, our approach outperforms the baseline with lower utilization of the communication links. In Fig.~\ref{fig:snapshot}, we take a snapshot of the risk-averse device selection and device participation strategies at $t=10$ to describe the operation of the proposed algorithm. In this scenario, we consider the availability of $9$ devices working in tandem with heatmaps reflecting contributions evaluated incorporating risk-averse participation and devices are selected accordingly for $M_t=5$ (represented by the ones enclosed in red boxes). Following the feedback, we observe the strategic participation of devices (Alg. 1), where only a subset of devices (represented by red blocks) transmit their local model parameters to \cs. hence, we achieve greater flexibility compared to the naive UCB algorithm, which involves a fixed set of devices selected based on their marginal contribution only and their subsequent participation. Our approach considers the risk-averse participation of devices, considering the influence of instantaneous transmission of local model parameters in improving the global model performance. This suggests a potential reduction in the number of irrelevant transmissions that would otherwise contribute to higher cumulative regrets.

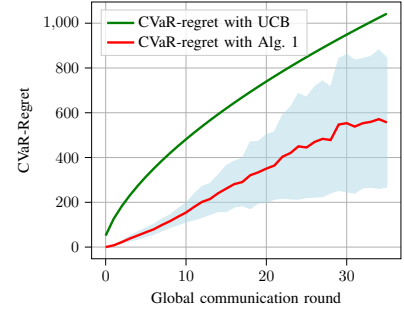
\begin{figure}[t!]
    \centering
\begin{tikzpicture}[scale=0.6]

\definecolor{darkgray176}{RGB}{176,176,176}
\definecolor{green01270}{RGB}{0,127,0}
\definecolor{lightblue}{RGB}{173,216,230}
\definecolor{lightgray204}{RGB}{204,204,204}

\begin{axis}[
legend cell align={left},
legend style={
  fill opacity=0.8,
  draw opacity=1,
  text opacity=1,
  at={(0.03,0.97)},
  anchor=north west,
  draw=lightgray204
},
tick align=outside,
tick pos=left,
x grid style={darkgray176},
xlabel={Global communication round},
xmajorgrids,
xmin=-1.75, xmax=36.75,
xtick style={color=black},
y grid style={darkgray176},
ylabel={CVaR-Regret},
ymajorgrids,
ymin=-52.1307186072728, ymax=1094.74509075273,
ytick style={color=black}
]
\path [draw=lightblue, fill=lightblue, opacity=0.5]
(axis cs:0,0)
--(axis cs:0,0)
--(axis cs:1,6.17655858668271)
--(axis cs:2,16.2463629300566)
--(axis cs:3,26.6969028710633)
--(axis cs:4,35.8594989030181)
--(axis cs:5,46.5547768339916)
--(axis cs:6,57.1669866764394)
--(axis cs:7,72.3051471236469)
--(axis cs:8,87.2249447299569)
--(axis cs:9,99.4049636072842)
--(axis cs:10,112.163275485104)
--(axis cs:11,121.587810749739)
--(axis cs:12,132.615513630196)
--(axis cs:13,145.418889089182)
--(axis cs:14,158.592383956229)
--(axis cs:15,158.423781159478)
--(axis cs:16,178.466441694242)
--(axis cs:17,184.651459536586)
--(axis cs:18,171.98411950335)
--(axis cs:19,197.054586520797)
--(axis cs:20,198.392937187197)
--(axis cs:21,213.569631674062)
--(axis cs:22,217.661076868886)
--(axis cs:23,214.858167734891)
--(axis cs:24,213.918871699373)
--(axis cs:25,221.128291394047)
--(axis cs:26,223.142464609979)
--(axis cs:27,224.846641142181)
--(axis cs:28,241.65443522702)
--(axis cs:29,253.192425031069)
--(axis cs:30,246.435587133551)
--(axis cs:31,240.812123383369)
--(axis cs:32,264.173474477439)
--(axis cs:33,267.277524525826)
--(axis cs:34,261.895581338772)
--(axis cs:35,268.470542057043)
--(axis cs:35,845.170503949918)
--(axis cs:35,845.170503949918)
--(axis cs:34,880.84057233312)
--(axis cs:33,851.589080143343)
--(axis cs:32,842.141615071723)
--(axis cs:31,835.525112036412)
--(axis cs:30,860.549504314198)
--(axis cs:29,841.567662164612)
--(axis cs:28,714.787530245177)
--(axis cs:27,741.964011693454)
--(axis cs:26,715.980026576312)
--(axis cs:25,668.655138415136)
--(axis cs:24,686.029642629301)
--(axis cs:23,625.389171219423)
--(axis cs:22,589.289971089361)
--(axis cs:21,514.85617251398)
--(axis cs:20,502.706839809535)
--(axis cs:19,471.102942099892)
--(axis cs:18,469.777064226884)
--(axis cs:17,396.091578620114)
--(axis cs:16,383.841183815272)
--(axis cs:15,363.58553032484)
--(axis cs:14,322.579562145558)
--(axis cs:13,284.584630459211)
--(axis cs:12,271.113574520577)
--(axis cs:11,236.929715217361)
--(axis cs:10,196.457942828555)
--(axis cs:9,171.716491298751)
--(axis cs:8,143.971564001141)
--(axis cs:7,125.230326754508)
--(axis cs:6,102.142828146241)
--(axis cs:5,85.5520943127769)
--(axis cs:4,67.9256362188427)
--(axis cs:3,49.3749953269731)
--(axis cs:2,28.0980947277002)
--(axis cs:1,9.55578749024886)
--(axis cs:0,0)
--cycle;

\addplot [line width=1.5, green01270]
table {%
0 51.6022694155895
1 126.399229637821
2 182.512893888081
3 230.772364413033
4 274.120663008581
5 313.967264817325
6 351.134272696964
7 386.156025066279
8 419.405043761825
9 451.154252948268
10 481.611186390888
11 510.938301329899
12 539.26576978067
13 566.699914215934
14 593.328983540509
15 619.227232987074
16 644.45788268603
17 669.075311892409
18 693.126718355196
19 716.653394787315
20 739.691725686769
21 762.273976272703
22 784.42892442803
23 806.182372392279
24 827.557565164902
25 848.575535691378
26 869.255391975532
27 889.614557681967
28 909.668975157416
29 929.433277835893
30 948.92093751193
31 968.144390838143
32 987.115148535492
33 1005.84389013076
34 1024.34054650815
35 1042.61437214546
};
\addlegendentry{CVaR-regret with UCB}
\addplot [line width=1.5, red]
table {%
0 0
1 7.86617303846578
2 22.1722288288784
3 38.0359490990182
4 51.8925675609304
5 66.0534355733843
6 79.6549074113402
7 98.7677369390774
8 115.598254365549
9 135.560727453018
10 154.31060915683
11 179.25876298355
12 201.864544075387
13 215.001759774197
14 240.585973050894
15 261.004655742159
16 281.153812754757
17 290.37151907835
18 320.880591865117
19 334.078764310344
20 350.549888498366
21 364.212902094021
22 403.475523979124
23 420.123669477157
24 449.974257164337
25 444.891714904591
26 469.561245593145
27 483.405326417817
28 478.220982736099
29 547.38004359784
30 553.492545723875
31 538.168617709891
32 553.157544774581
33 559.433302334584
34 571.368076835946
35 556.82052300348
};
\addlegendentry{CVaR-regret with Alg. 1}
\end{axis}

\end{tikzpicture}
\caption{CVaR-regret with UCB following strategic selection and risk-averse participation for an average risk of $0.9$ per
device. }
\label{fig:regret}
\end{figure}
Finally, in Fig.~\ref{fig:regret}, we show the CVaR-regret of our algorithm as compared to optimistic estimates, as shown in \eqref{eq:cvar_regret}. We observe the risk-averse participation approach offers lower CVaR-regret as compared to the worst-case scenario, offered for $90\%$ accuracy on a non-i.i.d. setting. This suggests a competitive sub-linear performance of our approach, offering better convergence benefits.

	\section{Conclusion and Future Work}
We evaluated the relevance of local model parameters for goal-oriented communications in FL through risk-averse profiling of participating devices. The proposed approach incorporates feedback to stir the transmission strategies of distributed devices towards a common goal of training a communication-efficient, better-generalized, high-quality FL model, hence, improving the overall goal-oriented performance.

In this work, we offer a goal-oriented perspective on semantic communication with data valuation for edge intelligence. Our early results suggest how fundamental it is to develop theoretical and algorithmic methodologies that capture the relevance of data prior to its transmission, to realize goal-oriented communications. While it is yet a challenge to ensure fully distributed coordination of transmission strategies, our future works will focus on timeliness in exploiting efficient data valuation techniques to anchor such attributes of semantic communications. Another potential direction of our interest is to devise techniques to manage the scale and utility of participation of devices in the networked real-time systems for goal-oriented tasks.

	\bibliographystyle{ieeetr}
	\bibliography{ref}

\end{document}